\documentclass{amsart}
\usepackage[ruled,vlined,linesnumbered]{algorithm2e}
\usepackage{amsmath}
\usepackage{amssymb}
\usepackage{tikz}
\usepackage{stmaryrd}

\newtheorem{thm}{Theorem}
\newtheorem{prop}[thm]{Proposition}

\theoremstyle{remark}

\newcommand{\assign}{\leftarrow}
\newcommand{\ep}{\varepsilon}
\newcommand{\ZZ}{\mathbf{Z}}
\newcommand{\CC}{\mathbf{C}}
\newcommand{\CCxx}{\CC \llbracket x \rrbracket}
\newcommand{\fourier}{\mathcal{F}}

\SetKwInOut{Output}{Output}
\SetKw{KwAnd}{and}

\begin{document}

\title{Faster exponentials of power series}
\author{David Harvey}

\begin{abstract}
We describe a new algorithm for computing $\exp f$ where $f$ is a power series in $\CCxx$. If $M(n)$ denotes the cost of multiplying polynomials of degree $n$, the new algorithm costs $(2.1666\ldots + o(1)) M(n)$ to compute $\exp f$ to order $n$. This improves on the previous best result, namely $(2.333\ldots + o(1)) M(n)$.
\end{abstract}

\maketitle

The author recently gave new algorithms for computing the square root and reciprocal of power series in $\CCxx$, achieving better running time constants than those previously known \cite{fast-series}. In this paper we apply similar techniques to the problem of computing $\exp f$ for a power series $f \in \CCxx$. Previously, the best known algorithm was that of van der Hoeven \cite[p.~6]{vdh}, computing $g = \exp(f) \bmod x^n$ in time $(7/3 + o(1)) M(n)$, where $M(n)$ denotes the cost of multiplying polynomials of degree $n$. We give a new algorithm that performs the same task in time $(13/6 + o(1)) M(n)$.

Van der Hoeven's algorithm works by decomposing $f$ into blocks, and solving $g' = f'g$ by operating systematically with FFTs of blocks. Our starting point is the observation that his algorithm computes \emph{too much}, in the sense that at the end of the algorithm, the FFT of every block of $g$ is known. Our new algorithm uses van der Hoeven's algorithm to compute the first half of $g$, and then extends the approximation to the target precision using a Newton iteration due to Brent \cite{brent} (see also \cite{hz} or \cite{bernstein} for other exponential algorithms based on a similar iteration). At the end of the algorithm, only the FFTs of the blocks of the first half of $g$ are known. In fact, the reduction in running time relative to van der Hoeven's algorithm turns out to be equal to the cost of these `missing' FFTs.

We freely use notation and complexity assumptions introduced in \cite{fast-series}. Briefly: `running time' always means number of ring operations in $\CC$. The Fourier transform of length $n$ is denoted by $\fourier_n(g)$, and its cost by $T(n)$. We assume that $T(2n) = (1/3 + o(1)) M(n)$ for a sufficiently dense set of integers $n$. For Proposition \ref{prop:exp} below, we fix a block size $m$, and for any $f \in \CCxx$ we write $f = f_{[0]} + f_{[1]} X + f_{[2]} X^2 + \cdots$ where $X = x^m$ and $\deg f_{[i]} < m$. The key technical tool is \cite[Lemma 1]{fast-series}, which asserts that if $f, g \in \CCxx$, $k \geq 0$, and if $\fourier_{2m}(f_{[i]})$ and $\fourier_{2m}(g_{[i]})$ are known for $0 \leq i \leq k$, then $(fg)_{[k]}$ may be computed in time $T(2m) + O(m(k+1))$.

We define a differential operator $\delta$ by $\delta f = x f'(x)$, and we set $\delta_k f = X^{-k} \delta (X^k f)$. In particular $\delta (f_{[0]} + f_{[1]}X + \cdots) = (\delta_0 f_{[0]}) + (\delta_1 f_{[1]}) X + \cdots$.

\begin{algorithm}
\label{algo:exp}
\dontprintsemicolon
\KwIn{$s \in \ZZ$, $s \geq 1$ \newline
$f \in \CCxx$, $f = 0 \bmod x$ \newline
$g_{[0]} = \exp(f_{[0]}) \bmod X$ \newline
$u = \exp(-f_{[0]}) \bmod X$ ($= g_{[0]}^{-1} \bmod X$)
}\;
\KwOut{$g = g_{[0]} + \cdots + g_{[2s-1]} X^{2s-1} = \exp(f) \bmod X^{2s}$}\;
\medskip\;
Compute $\fourier_{2m}(g_{[0]})$, $\fourier_{2m}(u)$\;
\lFor{$0 \leq k < s$}{compute $\fourier_{2m}((\delta f)_{[k]})}$\nllabel{line:delta-f}\;
\For{$1 \leq k < s$}{\nllabel{line:exp-1}
   $\psi \assign ((g_{[0]} + \cdots + g_{[k-1]} X^{k-1})((\delta f)_{[0]} + \cdots + (\delta f)_{[k]} X^k))_{[k]}$\nllabel{line:psi-1}\;
   Compute $\fourier_{2m}(\psi)$\;
   $\phi \assign u \psi \bmod X$\nllabel{line:phi}\;
   Compute $\fourier_{2m}(\delta_k^{-1} \phi)$\;
   $g_{[k]} \assign g_{[0]} (\delta_k^{-1} \phi) \bmod X$\nllabel{line:g}\;
   Compute $\fourier_{2m}(g_{[k]})$\nllabel{line:exp-2}\;
}\;
\lFor{$0 \leq k < s$}{\nllabel{line:quot-0}
   $q_{[k]} \assign (\delta f)_{[k]}$\;
}\;
\For{$s \leq k < 2s$}{\nllabel{line:quot-1}
   $\psi \assign ((q_{[0]} + \cdots + q_{[k-1]} X^{k-1})(g_{[0]} + \cdots + g_{[s-1]} X^{s-1}))_{[k]}$\nllabel{line:quot-mul}\;
   Compute $\fourier_{2m}(\psi)$\;
   $q_{[k]} \assign -u \psi \bmod X$\nllabel{line:q}\;
   Compute $\fourier_{2m}(q_{[k]})$\nllabel{line:quot-2}\;
}\;
\lFor{$0 \leq k < s$}{$\ep_{[k]} \assign \delta_{k+s}^{-1} q_{[k+s]} - f_{[k+s]}$}\nllabel{line:ep}\;
\lFor{$0 \leq k < s$}{compute $\fourier_{2m}(\ep_{[k]})}$\;
\lFor{$0 \leq k < s$}{$g_{[k+s]} \assign -((g_{[0]} + \cdots + g_{[k-1]} X^{k-1})(\ep_{[0]} + \cdots + \ep_{[k]} X^k))_{[k]}$}\nllabel{line:exp-f}\;
\caption{Exponential}
\end{algorithm}

\newpage 
\begin{prop}
\label{prop:exp}
Algorithm \ref{algo:exp} is correct, and runs in time $(13s - 4) T(2m) + O(s^2 m)$.
\end{prop}
\begin{proof}
We first show that the loop in lines \ref{line:exp-1}--\ref{line:exp-2} (essentially van der Hoeven's exponential algorithm) correctly computes
 \[ g_0 = g \bmod X^s = g_{[0]} + \cdots + g_{[s-1]} X^{s-1} = \exp(f) \bmod X^s. \]
By definition $g_{[0]}$ is correct. In the $k$th iteration, assume that $g_{[0]}, \ldots, g_{[k-1]}$ have been computed correctly. Since $\delta g_0 = g_0 (\delta f) \bmod X^s$ we have
  \[ ((g_{[0]} + \cdots + g_{[k]} X^k)((\delta f)_{[0]} + \cdots + (\delta f)_{[k]} X^k))_{[k]} = (\delta g)_{[k]}, \]
and by construction
  \[ ((g_{[0]} + \cdots + g_{[k-1]} X^{k-1})((\delta f)_{[0]} + \cdots + (\delta f)_{[k]} X^k))_{[k]} = \psi. \]
Subtracting yields
 \[ (\delta g)_{[k]} - \psi =  g_{[k]} (\delta f)_{[0]} \bmod X, \]
and on multiplying by $u$ we obtain
\begin{align*}
 \phi = u \psi \bmod X & = (\delta g)_{[k]} u - g_{[k]} (\delta f)_{[0]} u \mod X \\
    & = (\delta g)_{[k]} u + g_{[k]} (\delta u) \mod X \\
    & = \delta_k (g_{[k]} u \bmod X)
\end{align*}
since $\delta u = -(\delta f)u \bmod X$. Therefore $g_{[k]}$ is computed correctly in line \ref{line:g}.

Next we show that lines \ref{line:quot-0}--\ref{line:quot-2} correctly compute
 \[ q = q_{[0]} + \cdots + q_{[2s-1]} X^{2s-1} = \frac{\delta g_0}{g_0} \bmod X^{2s}. \]
Since $\delta g_0 / g_0 \bmod X^s = \delta f \bmod X^s$, line \ref{line:quot-0} correctly computes $q_{[0]}, \ldots, q_{[s-1]}$. The loop in lines \ref{line:quot-1}--\ref{line:quot-2} computes $q_{[s]}, \ldots, q_{[2s-1]}$ using a similar strategy to the division algorithm in \cite[p.~6]{vdh}. Namely, in the $k$th iteration, assume that $q_{[0]}, \ldots, q_{[k-1]}$ are correct. Then
 \[ ((q_{[0]} + \cdots + q_{[k-1]} X^{k-1})(g_{[0]} + \cdots + g_{[s-1]} X^{s-1}))_{[k]} = \psi \]
and
 \[ ((q_{[0]} + \cdots + q_{[k]} X^k)(g_{[0]} + \cdots + g_{[s-1]} X^{s-1}))_{[k]} = (q g_0)_{[k]} = (\delta g_0)_{[k]} = 0 \]
since $\deg(\delta g_0) < sm$ and $k \geq s$. Subtracting, we obtain $g_{[0]} q_{[k]} = -\psi \bmod X$, so $q_{[k]}$ is computed correctly in line \ref{line:q}. (Note that the transforms of $q_{[0]}, \ldots, q_{[s-1]}$ used in line \ref{line:quot-mul} are already known, since they were computed in line \ref{line:delta-f}.)

At this stage we have
 \[ \frac{\delta g_0}{g_0} \bmod X^{2s} = q = \delta f + \delta (\ep X^s) \]
for some $\ep = \ep_{[0]} + \cdots + \ep_{[s-1]} X^{s-1}$. Line \ref{line:ep} computes the blocks of $\ep$. Then by logarithmic integration, we have
 \[ g_0 = \exp(f) \exp(\ep X^s) \bmod X^{2s}, \]
so
 \[ \exp(f) \bmod X^{2s} = g_0 \exp(-\ep X^s) \bmod X^{2s} = g_0 (1 - \ep X^s) \bmod X^{2s}. \]
Line \ref{line:exp-f} multiplies out the latter product to compute the remaining blocks of $g$.

We now analyse the complexity. Each iteration of lines \ref{line:psi-1}, \ref{line:quot-mul} and \ref{line:exp-f} costs $T(2m) + O(m(k+1))$ according to \cite[Lemma 1]{fast-series}; their total contribution is therefore $(3s-1) T(2m) + O(s^2 m)$. Lines \ref{line:phi}, \ref{line:g} and \ref{line:q} each require a single inverse transform, contributing a total of $(3s - 2) T(2m)$. The explicitly stated forward transforms contribute $(7s - 1) T(2m)$. The various other operations, including applications of $\delta$ and $\delta^{-1}$, contribute only $O(sm)$. The total is $(13s - 4) T(2m) + O(s^2 m)$.
\end{proof}

\begin{thm}
\label{thm:exp}
Let $f \in \CCxx$ with $f = 0 \bmod x$. Then $\exp f$ may be computed to order $n$ in time $(13/6 + o(1)) M(n)$.
\end{thm}
\begin{proof}
Apply the proof of \cite[Theorem 3]{fast-series} to Proposition \ref{prop:exp}, with $r = 2s$.
\end{proof}

\bibliographystyle{amsalpha}
\bibliography{fast-exp}

\end{document}